\documentclass[conference]{IEEEtran}
\usepackage{amsmath,amssymb,amsthm,mathtools,bm}
\usepackage{graphicx,booktabs,xcolor,array}
\usepackage[hidelinks]{hyperref}
\usepackage{cite}
\usepackage[acronym,shortcuts]{glossaries}
\usepackage{balance}
\usepackage{float}
\graphicspath{{../figures/}}

\newtheorem{theorem}{Theorem}
\newtheorem{lemma}{Lemma}
\newtheorem{corollary}{Corollary}
\theoremstyle{definition}
\newtheorem{assumption}{Assumption}
\newtheorem{definition}{Definition}
\theoremstyle{remark}
\newtheorem{remark}{Remark}

\newcommand{\R}{\mathbb{R}}
\newcommand{\E}{\mathbb{E}}
\newcommand{\av}[1]{\langle #1\rangle}
\newcommand{\norm}[1]{\lVert #1\rVert}
\newcommand{\bH}{\mathbf{H}}
\newcommand{\by}{\mathbf{y}}
\newcommand{\bx}{\mathbf{x}}
\newcommand{\bz}{\mathbf{z}}
\newcommand{\br}{\mathbf{r}}
\newcommand{\bu}{\mathbf{u}}
\newcommand{\bv}{\mathbf{v}}
\newcommand{\bxb}{\bar{\mathbf{x}}}
\newcommand{\sbar}{\bar\sigma}
\newcommand{\zt}{\tilde z}
\newcommand{\cJ}{\mathcal J}
\newcommand{\cW}{\mathcal W}

\newacronym{AMP}{AMP}{approximate message passing}
\newacronym{SE}{SE}{state evolution}
\newacronym{MIMO}{MIMO}{multiple-input multiple-output}
\newacronym{CDMA}{CDMA}{code-division multiple access}
\newacronym{PAM}{PAM}{pulse amplitude modulation}
\newacronym{MSE}{MSE}{mean square error}
\newacronym{SER}{SER}{symbol error rate}
\newacronym{SNR}{SNR}{signal-to-noise ratio}
\newacronym{iid}{i.i.d.}{independent and identically distributed}

\begin{document}

\title{On the State Evolution of Approximate Message Passing with Discontinuous Denoisers
\vspace{-0.5ex}}

\author{\IEEEauthorblockN{Kuranage Roche Rayan Ranasinghe$^\star$, Takumi Takahashi$^\dag$, and Giuseppe Thadeu Freitas de Abreu$^\star$}
\IEEEauthorblockA{\emph{$^\star$School of Computer Science and Engineering, Constructor University, Bremen, Germany}\\
\emph{$^\dag$Graduate School of Engineering, Osaka University, Suita, Japan}\\
Emails: \{kranasinghe,gabreu\}@constructor.university, takahashi@comm.eng.osaka-u.ac.jp}\vspace{-4ex}}

\maketitle

\begin{abstract}
We consider \ac{AMP} with discontinuous denoisers, such as hard slicers, quantizers and hard thresholding, for which the classical \ac{SE} analysis does not apply.
It is first shown that the standard Onsager coefficient; i.e., the average pointwise derivative of the denoiser, misses a boundary term, namely, the sum over the jumps of the jump size times the density of the denoiser input at the jump.
Then, for real \ac{iid} Gaussian matrices, a positive noise variance, a fixed number of iterations and a prior with finite second moment, we transfer the Lipschitz \ac{SE} theorem directly to denoisers, fixed in advance, that are continuously differentiable with bounded derivative except at finitely many jumps, by bounding the distance between the iterates of \ac{AMP} and of \ac{AMP} on Lipschitz ramps of the denoisers through the fraction of inputs that cross a jump.
It is thereby proven that \ac{AMP} follows \ac{SE}, also for test functions such as symbol errors, if its Onsager coefficient contains the boundary term, evaluated either along \ac{SE} or, as we propose, at the noise level estimated from the residual, with the density computed from the known prior.
Simulation results with 4- and 8-level \ac{PAM} demonstrate that \ac{AMP} with a hard slicer, whose standard coefficient is identically zero, fails with this coefficient, while with the proposed one, it follows \ac{SE} and comes close to Bayes-optimal \ac{AMP}.
\end{abstract}

\glsresetall

\section{Introduction}
{\Ac{AMP} \cite{donoho2009message,bayati2011dynamics} has become a standard tool in compressed sensing and in large \ac{MIMO} detection \cite{jeon2015optimality}.
Specifically, \ac{AMP} alternates a matched filter and a separable denoiser $\eta_t$, and adds an Onsager correction, scaled by a coefficient $b_t$ usually given by the empirical average of $\eta_t'$, with which, in the large-system limit, each entry of the denoiser input behaves as the output of a scalar Gaussian channel whose noise variance follows a deterministic recursion, referred to as \ac{SE} \cite{bayati2011dynamics,bolthausen2014iterative,javanmard2013state}.}

Most \ac{SE} proofs, however, require Lipschitz denoisers \cite{bayati2011dynamics,javanmard2013state,berthier2020state,gerbelot2023graph,rush2018finite,feng2022unifying}, which excludes hard decisions, quantizers and hard thresholding, for which failures of \ac{AMP} have in fact been reported: \ac{SE} predictions with hard thresholding ``fail entirely'' \cite{metzler2016denoising}, the divergence of a phase-retrieval denoiser built on the sign function contains a Dirac delta \cite{ma2019optimization,saglietti2020generalized}, and $\ell_0$-\ac{AMP} is unstable \cite{barbier2025phase}.
The remedies proposed so far have been either to smooth the denoiser, for which \ac{SE} holds in the iterated limit in which the smoothing vanishes after the system size has grown \cite{zheng2017does,ma2019optimization}, or to take the divergence from \ac{SE}, albeit without proof \cite{tanaka2005approximate,saglietti2020generalized,barbier2025phase}.

For hard decisions, the missing boundary term has a long heuristic history: in statistical neurodynamics, the noise-enhancement factor of the sign function is twice the Gaussian density of the neuron input at the threshold \cite{amari1988statistical,okada1995hierarchy}, with staircase functions yielding the step heights times the densities at the thresholds \cite{bolle1998parallel}, while Tanaka and Okada \cite{tanaka2005approximate} derived the same factor for interference cancellation in \ac{CDMA} with smooth decisions, computed it offline from density evolution, and expected hard decisions to follow by smoothing, as supported by their simulations and by the generating-functional analysis of \cite{mimura2014generating}.
These analyses are, however, heuristic, and Theorem~\ref{thm:se} proves, for \ac{AMP} and general denoisers with jumps, what \cite{tanaka2005approximate} expected for hard decisions, while the same coefficient appears in \cite{wu2025asymptotic} only for a single quantization step, analyzed by mollification.
Furthermore, \ac{SE} has been claimed for piecewise-Lipschitz denoisers with the almost-everywhere derivative in Stein's lemma \cite{takeuchi2024decentralized,takeuchi2025generalized,takeuchi2025goamp}, a form that misses the boundary term (see Lemma~\ref{lem:stein}), since almost-everywhere differentiability does not suffice for Stein's lemma \cite{hansen2018stein}.\footnote{For $\eta(u)=|u|+\mathbf 1\{u>0\}$, considered in \cite{takeuchi2024decentralized}, this form predicts $\E[\zt\eta(\zt)]=0$, whereas in fact $\E[\zt\eta(\zt)]=\varphi(0)$.}

We highlight that \ac{SE} without smoothness is known in the tensor-program framework \cite{yang2019scaling,yang2020tensor}, whose nonlinearities and test functions need only be polynomially bounded, but whose signal is Gaussian \cite[Cor.~D.18]{yang2019scaling} or a polynomially bounded function of a Gaussian vector \cite[Thm.~F.8]{yang2020tensor}; these results are compared with ours in Remark~\ref{rem:yang}.

In view of all the above, the contributions of this article are summarized as follows:
\begin{itemize}
\item Using the Lipschitz \ac{SE} of \cite{bayati2011dynamics} as a black box, and a crossing bound that keeps the iterates of \ac{AMP} close to those of \ac{AMP} on Lipschitz ramps, we prove in Section~\ref{sec:theorem} that, for real \ac{iid} Gaussian matrices, $\sigma_z^2>0$, a fixed $T$, a prior with finite second moment and denoisers fixed in advance with finitely many jumps (Assumption~\ref{ass}), \ac{AMP} follows \ac{SE}, also for test functions such as symbol errors, if its Onsager coefficient {is the distributional or the \ac{SE}-corrected one, combining the expectation or empirical average of the pointwise derivative, respectively, with the corresponding boundary term} (Theorem~\ref{thm:se}).
\item We propose in Section~\ref{sec:principle} an \acs{SE}-corrected Onsager coefficient (Definition~\ref{def:bcorr}, Theorem~\ref{thm:se}(b)), which adds the boundary term at the noise level estimated from the residual to the empirical average of the pointwise derivative, and requires the known prior, but neither the \ac{SE} trajectory nor the signal, at a negligible cost.
\item Finally, we show that the standard coefficient is inconsistent whenever the boundary term is non-zero, as it is for every slicer, whose correct coefficient is the multilevel form of the heuristic binary factor \cite{tanaka2005approximate,mimura2014generating} (Corollary~\ref{cor:consequences}), and confirm in Section~\ref{sec:numerics} that slicer-\ac{AMP}; i.e., \ac{AMP} with a hard slicer, comes close to Bayes-optimal \ac{AMP} with the proposed coefficient.
\end{itemize}

\emph{Notation:} Boldface lower- and uppercase letters denote vectors and matrices, respectively; $\norm{\cdot}$ and $\norm{\cdot}_{\rm op}$ are the Euclidean and spectral norms; $\av{\bu}\triangleq\frac1K\sum_{k=1}^Ku_k$ is the empirical average of $\bu\in\R^K$, and scalar functions act entrywise; $\mathcal N(m,v)$ is the Gaussian distribution with mean $m$ and variance $v$; $\zt\sim\mathcal N(0,1)$ is independent of everything else, with density $\varphi$; $\mathbf 1\{\cdot\}$ is the indicator function; $\eta(\beta^\pm)$ are the one-sided limits of $\eta:\R\to\R$ at $\beta$, and $\eta'_{\rm ae}$ is its derivative where it exists and $0$ elsewhere; $\mathcal D_s\triangleq\{u:\mathrm{dist}(u,\mathcal D)\le s\}$, with $\mathrm{dist}(u,\mathcal D)\triangleq\inf_{d\in\mathcal D}|u-d|$; $\nu_{\bu}(\mathcal A)\triangleq\frac1K\#\{k:u_k\in\mathcal A\}$, where $\#$ counts elements; and $\mathrm{supp}\,p$ is the support of $p$.

\section{AMP and State Evolution}\label{sec:model}
We consider $N$ noisy linear measurements $\by=\bH\bx+\bz\in\R^N$ of $K$ unknowns $\bx\in\R^K$, where $H_{nk}\sim\mathcal N(0,1/N)$, the unknowns $x_k\sim p$ are \ac{iid} with $E_d\triangleq\E[x_k^2]<\infty$, and $\bz\sim\mathcal N(\mathbf 0,\sigma_z^2\mathbf I)$ with $\sigma_z>0$, all mutually independent.
The load $\alpha\triangleq K/N\in(0,\infty)$ is fixed as $K,N\to\infty$, which we hereafter write as $K\to\infty$, and all problem sizes are defined on one probability space,\footnote{Specifically, $\bx$ and $\bz$ are the first $K$ and $N$ entries of two \ac{iid} sequences, and $\bH=\mathbf G/\sqrt N$, where $\mathbf G$ is the upper-left $N\times K$ block of an infinite array of \ac{iid} $\mathcal N(0,1)$ entries.} to which ``almost surely'' (a.s.) refers.

Given a number of iterations $T$, denoisers $\eta_t:\R\to\R$ and coefficients $b_t$, \ac{AMP} starts at $\bx^0=\br^0=\mathbf 0$, $b_0=0$, and computes, for $t=1,\dots,T$,
\begin{equation}\label{eq:amp}
\begin{aligned}
\br^t&=\by-\bH\bx^{t-1}+\alpha b_{t-1}\br^{t-1},\\
\bxb^t&=\bx^{t-1}+\bH^{\mathsf T}\br^t,\qquad\bx^t=\eta_t(\bxb^t),
\end{aligned}
\end{equation}
with residual $\br^t$ and cavity $\bxb^t$, and with $\br^{T+1}$ and $\bxb^{T+1}$ obtained from the first two steps.

The standard Onsager coefficient is $b_t=\av{\eta'_{t,{\rm ae}}(\bxb^t)}$, and since the denoisers are fixed functions, the noise estimate $\hat\sigma^2_t\triangleq\norm{\br^t}^2/N$ enters only the coefficient of equation~\eqref{eq:bcorr}.
In turn, \ac{SE} is the scalar recursion
\begin{equation}\label{eq:se}
\sbar^2_{t+1}=\sigma_z^2+\alpha\,\E\big[(\eta_t(x+\sbar_t\zt)-x)^2\big],\quad \sbar^2_1=\sigma_z^2+\alpha E_d,
\end{equation}
with $x\sim p$, such that $\sbar_t\ge\sigma_z$.
Let $\mu_t$ denote the law of $(x,x+\sbar_t\zt)$, and recall that a function $\psi:\R^2\to\R$ is pseudo-Lipschitz of order two (PL2) if $|\psi(\mathbf a)-\psi(\mathbf a')|\le C(1+\|\mathbf a\|+\|\mathbf a'\|)\|\mathbf a-\mathbf a'\|$ for some $C<\infty$.
Then, \ac{SE} is said to describe \ac{AMP} at iteration $t$ if, a.s., $\hat\sigma_t^2\to\sbar_t^2$ and, for all PL2 $\psi$ simultaneously,
\begin{equation}\label{eq:se_limit}
\lim_{K\to\infty}\frac1K\sum_{k=1}^K\psi(x_k,\bar x^t_k)=\E\big[\psi(x,x+\sbar_t\zt)\big],
\end{equation}
or, equivalently, if we have that the empirical distribution of $\{(x_k,\bar x^t_k)\}_{k=1}^K$ converges to $\mu_t$ in Wasserstein-2 distance.

\begin{lemma}[{[\citen{bayati2011dynamics}, Thm.~2]}]\label{lem:bm}
Let $E_d>0$, let each $\eta_t$ be Lipschitz and continuously differentiable outside a finite set; e.g., soft thresholding or the posterior mean for a \ac{PAM} prior, and let $b_t=\av{\eta'_{t,{\rm ae}}(\bxb^t)}$. Then \ac{SE} describes \ac{AMP} at every $t\le T+1$, and $b_t\to\E[\eta'_{t,{\rm ae}}(x+\sbar_t\zt)]$ a.s. for $t\le T$.
\end{lemma}
\begin{proof}
Applying \cite[Thm.~2]{bayati2011dynamics} through the reduction of \cite[Sec.~III-C]{bayati2011dynamics}, conditionally on $(\bx,\bz)$, whose empirical distributions and second moments converge a.s., and with $E_d>0$ as required by \cite[eq.~(3.4)]{bayati2011dynamics}, yields equation~\eqref{eq:se_limit} a.s. for each $\psi$, and thus on one event for all PL2 $\psi$, since Wasserstein-2 convergence on $\R^2$ follows from equation~\eqref{eq:se_limit} for countably many bounded Lipschitz $\psi$ and for $\psi(\mathbf a)=\|\mathbf a\|^2$.
The limit of $b_t$ then follows from the mapping theorem \cite[Thm.~2.7]{billingsley1999convergence}, since $\eta'_{t,{\rm ae}}$ is bounded and continuous outside a finite set, which has probability zero under the limit.
\end{proof}

\section{The Distributional Onsager Coefficient}\label{sec:principle}
We first recall Stein's identity for discontinuous functions.

\begin{lemma}[{[\citen{tibshirani2015degrees}, Lemma~4]}]\label{lem:stein}
Let $\eta:\R\to\R$ be Lipschitz on each open interval between the points $\beta_1<\dots<\beta_m$, with jumps $\Delta_j=\eta(\beta_j^+)-\eta(\beta_j^-)$, and let $\bar x=x+\sigma\zt$ with $\sigma>0$ and $x\sim p$. Then $\bar x$ has the continuous, positive density $f_\sigma(u)=\E[\varphi((u-x)/\sigma)]/\sigma\le1/(\sqrt{2\pi}\,\sigma)$, and
\begin{equation}\label{eq:stein}
\bar b(\sigma)\triangleq\frac{\E[\zt\,\eta(\bar x)]}{\sigma}=\E\big[\eta'_{\rm ae}(\bar x)\big]+\sum_{j=1}^{m}\Delta_j\,f_\sigma(\beta_j).
\end{equation}
\end{lemma}
\begin{proof}
Fix $x$ and let $\gamma_j\triangleq(\beta_j-x)/\sigma$.
On each interval between consecutive $\gamma_j$, $z\mapsto\eta(x+\sigma z)$ is absolutely continuous, such that integration by parts with $z\varphi(z)=-\varphi'(z)$ yields $\sigma\E[\eta'_{\rm ae}(\bar x)\mid x]$ plus the boundary values $\eta(\beta_j^+)\varphi(\gamma_j)-\eta(\beta_j^-)\varphi(\gamma_j)$, where the terms at $\pm\infty$ vanish since $\eta$ grows at most linearly.
Averaging over $x$ then turns $\varphi(\gamma_j)$ into $\sigma f_\sigma(\beta_j)$.
\end{proof}

We refer to $\bar b(\sigma)$ as the \emph{distributional Onsager coefficient}, since equation~\eqref{eq:stein} integrates the distributional derivative of $\eta$, namely $\eta'_{\rm ae}$ plus point masses $\Delta_j$ at the $\beta_j$, against $f_\sigma$, and write $\bar b_t\triangleq\bar b(\sbar_t)$ for $\eta=\eta_t$.
Its second term, $B(\sigma)\triangleq\sum_j\Delta_jf_\sigma(\beta_j)$, is hereafter referred to as the \emph{boundary term}, with $B_t\triangleq B(\sbar_t)$.
Consider, for instance, the hard slicer onto a \ac{PAM} alphabet $\mathcal M=\{a_1<\dots<a_M\}$; i.e., $\eta(u)=\arg\min_{a\in\mathcal M}|u-a|$, whose jumps $\Delta_i=a_{i+1}-a_i$ are located at the thresholds $\beta_i=(a_i+a_{i+1})/2$, and whose pointwise derivative is $\eta'_{\rm ae}=0$.
The standard coefficient is then zero, such that \ac{AMP} reduces to interference cancellation without Onsager correction, whereas $\bar b(\sigma)=\sum_{i=1}^{M-1}(a_{i+1}-a_i)f_\sigma(\beta_i)>0$, which for $\{\pm1\}$ with equal probabilities reduces to $\bar b(\sigma)=2\varphi(1/\sigma)/\sigma$; i.e., the factor of \cite{amari1988statistical,okada1995hierarchy,tanaka2005approximate}.
Since $\bar b_t$ requires the \ac{SE} trajectory, computed offline as in \cite{tanaka2005approximate}, we propose the following coefficient, which \ac{AMP} computes online.

\begin{table}[H]
\caption{Onsager coefficients for discontinuous denoisers}\label{tab:coef}
\vspace{-2ex}
\centering\footnotesize
\setlength{\tabcolsep}{3pt}
\begin{tabular}{@{}l|l|l|>{\raggedright\arraybackslash}p{3.3cm}@{}}
\toprule
Coefficient & Formula & Computed from & Result\\
\midrule
standard & $\av{\eta'_{t,{\rm ae}}(\bxb^t)}$ & $\bxb^t$, online & inconsist. if $B_t\neq0$ (Cor.~\ref{cor:consequences})\\
distributional & $\bar b_t$, eq.~\eqref{eq:stein} & $\sbar_t$, offline & \acs{SE} holds (Thm.~\ref{thm:se}(a))\\
\acs{SE}-corrected & $\hat b_t$, eq.~\eqref{eq:bcorr} & $\bxb^t$, $\br^t$, $p$, online & \acs{SE} holds (Thm.~\ref{thm:se}(b))\\
\bottomrule
\end{tabular}
\vspace{-1ex}
\end{table}

\begin{definition}[\acs{SE}-corrected Onsager coefficient]\label{def:bcorr}
Let $\eta_t$ be as in Lemma~\ref{lem:stein}, with $J_t$ points $\beta_{t,j}$ and jumps $\Delta_{t,j}$. The \emph{\acs{SE}-corrected Onsager coefficient} is defined as
\begin{equation}\label{eq:bcorr}
\hat b_t\triangleq\av{\eta'_{t,{\rm ae}}(\bxb^t)}+\sum_{j=1}^{J_t}\Delta_{t,j}\,f_{\hat\sigma_t}(\beta_{t,j}),\qquad \hat\sigma_t^2=\frac{\norm{\br^t}^2}N,
\end{equation}
where $f_{\hat\sigma_t}$ is the density of Lemma~\ref{lem:stein}, computed from the known prior $p$, and the boundary term is set to $0$ if $\hat\sigma_t=0$.
\end{definition}

In equation~\eqref{eq:bcorr}, the boundary term is thus evaluated at $\hat\sigma_t$ instead of $\sbar_t$, such that \ac{SE} enters only through the Gaussian form of the cavity, hence the name.\footnote{Smoothing, in contrast, hides the boundary term rather than removing it: if each jump is replaced by a linear ramp of width $2\epsilon$, the standard coefficient of the ramped denoiser contains $(\Delta_j+\mathcal O(\epsilon))\,\nu_{\bxb^t}((\beta_j-\epsilon,\beta_j+\epsilon))/(2\epsilon)$; i.e., a histogram estimate of the cavity density at $\beta_j$, as in DoF estimators for hard thresholding \cite{deledalle2013score}, whereas equation~\eqref{eq:bcorr} uses the \ac{SE} density instead.}
We emphasize that the prior must be known (as for the Bayes-optimal denoiser), since with another distribution in place of $p$, the limit of equation~\eqref{eq:bcorr} differs from $\bar b_t$ in general.
For a prior with finitely many atoms, the boundary term costs $\mathcal O(J_t\,\#\mathrm{supp}\,p)$ operations and $\hat\sigma_t$ costs $\mathcal O(N)$, against $\mathcal O(NK)$ for the matrix--vector products.
For the slicer and the uniform prior on $\mathcal M$, in particular, equation~\eqref{eq:bcorr} is $\hat b_t=\frac1{M\hat\sigma_t}\sum_{i=1}^{M-1}\sum_{c=1}^M(a_{i+1}-a_i)\varphi((\beta_i-a_c)/\hat\sigma_t)$.
The three coefficients are summarized in Table~\ref{tab:coef}.

\section{State Evolution for Discontinuous Denoisers}\label{sec:theorem}
\begin{assumption}\label{ass}
There are $J\in\mathbb N$, $L,\Delta<\infty$ and $g>0$ such that:
(i) $E_d>0$;
(ii) for each $t\le T$, there are $J_t\le J$ points $\beta_{t,1}<\dots<\beta_{t,J_t}$ with gaps $\beta_{t,j+1}-\beta_{t,j}\ge g$, such that $\eta_t$ is continuously differentiable with $|\eta_t'|\le L$ on each open interval between them;
(iii) the jumps satisfy $|\Delta_{t,j}|\le\Delta$, and $\eta_t(\beta_{t,j})$ lies between $\eta_t(\beta_{t,j}^-)$ and $\eta_t(\beta_{t,j}^+)$.
\end{assumption}

Note that points with $\Delta_{t,j}=0$ are kinks, such that Assumption~\ref{ass} contains the denoisers of Lemma~\ref{lem:bm}, and also covers \ac{PAM} slicers with any tie-breaking rule, quantizers with finitely many levels and hard thresholding $\eta(u)=u\,\mathbf 1\{|u|>\lambda\}$ with fixed $\lambda>0$.
Our main result then reads as follows, with $\cJ_t\triangleq\{\beta_{t,1},\dots,\beta_{t,J_t}\}$ and its closed $s$-neighbourhood $\cJ_{t,s}$.

\begin{theorem}\label{thm:se}
Let Assumption~\ref{ass} hold, and run \ac{AMP} \eqref{eq:amp} with (a) $b_t=\bar b_t$ or (b) $b_t=\hat b_t$. Then, a.s., $b_t\to\bar b_t$ for $t\le T$, and \ac{SE} describes \ac{AMP} at every $t\le T+1$. Moreover, a.s., \eqref{eq:se_limit} holds at every $t\le T+1$ simultaneously for all Borel $\psi$ with $|\psi(\mathbf a)|\le C(1+\|\mathbf a\|^2)$ whose set of discontinuity points $D_\psi$ satisfies $\mu_t(D_\psi)=0$.
Finally, let $\mathcal E=\{1/m:m\in\mathbb N,\ m>2/g\}$, and, for $\epsilon\in\mathcal E$, let $\bxb^t_\epsilon$ and $\br^t_\epsilon$ be the iterates of \ac{AMP} run on the same $(\bH,\bx,\bz)$ with the Lipschitz ramps $\eta_{t,\epsilon}$ of Lemma~\ref{lem:ramp} and their standard coefficients $b_{t,\epsilon}=\av{\eta'_{t,\epsilon}(\bxb^t_\epsilon)}$.

Then, a.s.,
\begin{equation}\label{eq:close}
\lim_{\substack{\epsilon\downarrow0\\ \epsilon\in\mathcal E}}\,\limsup_{K\to\infty}\,\max_{1\le t\le T+1}\bigg[\frac{\norm{\bxb^t-\bxb^t_\epsilon}}{\sqrt K}+\frac{\norm{\br^t-\br^t_\epsilon}}{\sqrt N}\bigg]=0.
\end{equation}
\end{theorem}

Since $x+\sbar_t\zt$ has a density, the condition $\mu_t(D_\psi)=0$ holds whenever $D_\psi\subset(\mathcal A\times\R)\cup(\R\times\mathcal N)$, with $\Pr(x\in\mathcal A)=0$ and $\mathcal N$ of zero length; e.g., $D_\psi\subset\R\times\cJ_t$ for $\psi(x,\bar x)=(\eta_t(\bar x)-x)^2$.

\begin{corollary}\label{cor:consequences}
Under the assumptions of Theorem~\ref{thm:se}, a.s.: (i) for $t\le T$, $\frac1K\norm{\bx^t-\bx}^2\to\E[(\eta_t(x+\sbar_t\zt)-x)^2]$, and, if $\mathrm{supp}\,p\subset\mathcal M$ and $\eta$ is the slicer onto $\mathcal M$, the \ac{SER} $\frac1K\#\{k:\eta(\bar x^{T+1}_k)\neq x_k\}$ converges to $\Pr(\eta(x+\sbar_{T+1}\zt)\neq x)$; (ii) for $t\le T$, $\av{\eta'_{t,{\rm ae}}(\bxb^t)}\to\bar b_t-B_t$, such that the standard coefficient is inconsistent whenever $B_t\neq0$, where $B_t=\bar b_t>0$ for a slicer.
\end{corollary}
\begin{proof}
The result follows from Theorem~\ref{thm:se} applied to $(\eta_t(\bar x)-x)^2$, which grows quadratically and is discontinuous only on $\R\times\cJ_t$; to $\mathbf 1\{\eta(\bar x)\neq\eta(x)\}$, which equals the error indicator on $\mathrm{supp}\,p$ and is discontinuous only on $(\cJ\times\R)\cup(\R\times\cJ)$, a $\mu_{T+1}$-null set since the thresholds are not atoms of $p$; and to $\eta'_{t,{\rm ae}}(\bar x)$, whose limit is $\bar b_t-B_t$ by equation~\eqref{eq:stein}.
\end{proof}

We emphasize that no other recursion is claimed for the standard coefficient, which in fact fails for the slicer (see Section~\ref{sec:numerics}).

\begin{remark}[Relation to tensor programs]\label{rem:yang}
Theorem~\ref{thm:se} and \cite{yang2019scaling,yang2020tensor} overlap, but neither contains the other.
On the one hand, \cite[Cor.~D.18]{yang2019scaling} assumes Gaussian limits for the empirical distributions of the signal and the noise, mentioning priors with all moments finite as possible with more effort, and in \cite[Thm.~F.8]{yang2020tensor} a non-Gaussian signal; e.g., \ac{PAM}, enters as a polynomially bounded function of a Gaussian vector, in exchange for nonlinearities and test functions that need only be polynomially bounded.
On the other hand, Theorem~\ref{thm:se} requires of $p$ only $0<E_d<\infty$, and thus covers; e.g., heavy-tailed priors with finite variance.
In addition, through \cite[Sec.~III-C]{bayati2011dynamics}, the coefficient of \cite[Cor.~D.18]{yang2019scaling} reads $\av{(\bxb^t-\bx)(\eta_t(\bxb^t)-\bx)}/\sbar_t^2$ in our notation, which involves the unknown $\bx$ and $\sbar_t^2$ and converges to $\bar b_t$, while the Onsager term in the limit of \cite[Thm.~F.8]{yang2020tensor} is defined through Gaussian correlations \cite[Remark~F.7]{yang2020tensor}, which, for the denoisers of Assumption~\ref{ass}, equal the expectations of distributional derivatives in equation~\eqref{eq:stein} (cf.\ \cite[fn.~10]{yang2019scaling}).
Hence, case (a), with a coefficient fixed in advance, is within the scope of \cite[Thm.~F.8]{yang2020tensor} for such signals, whereas $\hat b_t$ depends on the iterates and falls under \cite[Thms.~F.9--F.10]{yang2020tensor}, whose parameter-control and rank-stability conditions we have not verified for equation~\eqref{eq:bcorr}.
Finally, our proof uses Lemma~\ref{lem:bm} as a black box instead of Gaussian conditioning, and yields equation~\eqref{eq:close}.
\end{remark}

\begin{remark}[Scope]
The noise floor $\sigma_z>0$, which bounds the \ac{SE} densities by $f_{\max}\triangleq(2\pi\sigma_z^2)^{-1/2}$, can be replaced by $\min_{t\le T+1}\sbar_t>0$, and no bounded support of $p$ is needed, such that Gaussian and Bernoulli--Gaussian priors with fixed thresholds are covered, but thresholds set from $\hat\sigma_t$, which make $\eta_t$ data-dependent, are not.
Since the proof uses $\bH$ only through Lemma~\ref{lem:bm} and $\norm{\bH}_{\rm op}\to1+\sqrt\alpha$ a.s. \cite{geman1980limit}, and $b_t$ only through $b_t\to\bar b_t$, the transfer applies wherever Lipschitz \ac{SE} holds as in Lemma~\ref{lem:bm} with bounded $\norm{\bH}_{\rm op}$, for any such coefficients.
Unlike the iterated limit for smoothed \ac{AMP} \cite{zheng2017does,ma2019optimization}, in which $K\to\infty$ before $\epsilon\to0$, Theorem~\ref{thm:se} describes the discontinuous iteration itself, for which the order of the limits in equation~\eqref{eq:close} matters: at fixed $K$, if no cavity entry lies on a jump point, \ac{AMP} with the standard coefficient coincides with its ramped version for all small enough $\epsilon$, and need not follow \ac{SE} (see Section~\ref{sec:numerics}).
\end{remark}

\begin{figure*}[!t]
    \vspace{1ex}
\centering
\includegraphics[width=\textwidth]{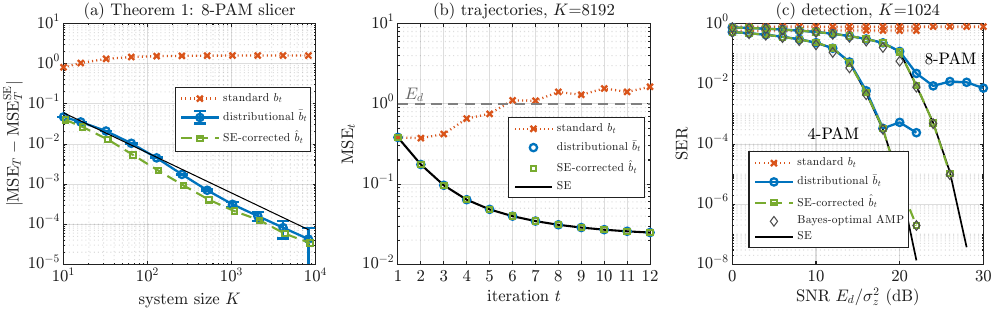}
\vspace{-4.5ex}
\caption{Slicer-\acs{AMP}, $\alpha=0.5$; markers and colours are the same in all panels. (a) Gap $|\mathrm{MSE}_T-\mathrm{MSE}_T^{\rm SE}|$ between the empirical \acs{MSE} $\frac1K\norm{\bx^T-\bx}^2$ after $T=12$ iterations, averaged over realizations, and its \acs{SE} value, against $K$ (8-\acs{PAM}, 20~dB, at least $2^{25}$ symbols per point; error bars: $\pm2$ standard errors of the $\bar b_t$ curve; $\hat b_t$ curve shifted by $5\%$ in $K$; thin line: $\propto K^{-1}$). (b) The same runs at $K=8192$: $\mathrm{MSE}_t$ against its \acs{SE} value. (c) \acs{SER} of the slicer decisions on $\bxb^{T+1}$ for 4- and 8-\acs{PAM} ($K=1024$, $N=2048$, $T=20$), and of Bayes-optimal \acs{AMP}; black lines: \acs{SE} of slicer-\acs{AMP}.}\label{fig:main}
\vspace{-3.5ex}
\end{figure*}

\section{Proof of Theorem~\ref{thm:se}}\label{sec:proof}
Since a small perturbation of a cavity entry near a jump point can change the denoiser output by up to $\Delta$, we compare \ac{AMP} with the same iteration on the Lipschitz ramps of Lemma~\ref{lem:ramp}, to which Lemma~\ref{lem:bm} applies.
The two iterations differ by $\mathcal O(1)$ only in entries that cross a jump, which the crossing bound of Lemma~\ref{lem:cross} controls by the fraction of cavity entries within distance $s$ of a jump point, which is $\mathcal O(s)$ since the \ac{SE} law of the cavity has a bounded density, plus a Chebyshev term, while their Onsager coefficients are close because the \ac{SE} coefficient of the ramp converges to $\bar b_t$, boundary term included (Step~1).

To that end, let constants $C,C_1,\dots$ in Steps~1--3 depend on $T,\alpha,\sigma_z,p,J,L,\Delta,g$ and the $\eta_t$, but not on $K$ or $\epsilon$.
For $\sigma\ge\sigma_z$, we have $f_\sigma\le f_{\max}$ by Lemma~\ref{lem:stein}, and since $\cJ_{t,s}$ is a union of at most $J$ intervals of length $2s$, it follows that
\begin{equation}\label{eq:tube}
\Pr(x+\sigma\zt\in\cJ_{t,s})\le2Jsf_{\max},\qquad s>0.
\end{equation}

\begin{lemma}[Crossing bound]\label{lem:cross}
Under Assumption~\ref{ass}, with $L'=L+\Delta/g$, for $t\le T$ and $u,v\in\R$,
\begin{equation}\label{eq:jumpbound}
|\eta_t(v)-\eta_t(u)|\le L'|v-u|+\Delta\,\mathbf 1\{\mathrm{dist}(u,\cJ_t)\le|v-u|\}.
\end{equation}
Hence, for $\bu,\bv\in\R^K$, $e^2=\frac1K\norm{\bv-\bu}^2$ and $s>0$,
\begin{equation}\label{eq:flip}
\tfrac1K\norm{\eta_t(\bv)-\eta_t(\bu)}^2\le2L'^2e^2+2\Delta^2\big[\nu_{\bu}(\cJ_{t,s})+e^2/s^2\big].
\end{equation}
\end{lemma}
\begin{proof}
On the closed interval $I$ between $u$ and $v$, $\eta_t$ is $L$-Lipschitz between jump points and, by Assumption~\ref{ass}(iii), varies by at most $|\Delta_{t,j}|$ at each jump point in $I$. If $I$ contains a jump point, then $\mathrm{dist}(u,\cJ_t)\le|v-u|$, and $I$ contains at most $1+|v-u|/g$ of them, which yields \eqref{eq:jumpbound}. In turn, \eqref{eq:flip} follows from squaring \eqref{eq:jumpbound}, using $\mathbf 1\{\mathrm{dist}(u_k,\cJ_t)\le|d_k|\}\le\mathbf 1\{u_k\in\cJ_{t,s}\}+d_k^2/s^2$ with $d_k=v_k-u_k$, and averaging over $k$.
\end{proof}

\begin{lemma}[Lipschitz ramps]\label{lem:ramp}
For $t\le T$ and $0<\epsilon<g/2$, let $\eta_{t,\epsilon}$ equal $\eta_t$ outside $\cW_{t,\epsilon}=\cJ_{t,\epsilon}$ and be affine on each interval of $\cW_{t,\epsilon}$, continuous at its endpoints. Then $\eta_{t,\epsilon}$ is Lipschitz with constant $L+\Delta/(2\epsilon)$ and continuously differentiable outside a finite set, and $|\eta_{t,\epsilon}-\eta_t|\le\Delta'=2\Delta+2Lg$ on $\cW_{t,\epsilon}$.
\end{lemma}
\begin{proof}
Each window $[\beta_{t,j}-\epsilon,\beta_{t,j}+\epsilon]$ contains one jump point, such that, as in the proof of Lemma~\ref{lem:cross}, the values of $\eta_t$ in it, including its endpoints, lie within $2L\epsilon+\Delta$ of $\eta_t(\beta_{t,j}-\epsilon)$. This bounds the slope of $\eta_{t,\epsilon}$ by $L+\Delta/(2\epsilon)$, and $|\eta_{t,\epsilon}-\eta_t|$ by $2\Delta+4L\epsilon\le\Delta'$.
\end{proof}

\begin{proof}[In view of the above Theorem~\ref{thm:se} can be proved as follows]\quad\\[-2.5ex]

\emph{Step 1 (\ac{SE} of the Ramps):}
Let $\sbar_{t,\epsilon}$ follow \eqref{eq:se} with $\eta_{t,\epsilon}$, and let $\bar b_{t,\epsilon}=\E[\eta'_{t,\epsilon}(x+\sbar_{t,\epsilon}\zt)]=\E[\zt\eta_{t,\epsilon}(x+\sbar_{t,\epsilon}\zt)]/\sbar_{t,\epsilon}$, by Lemma~\ref{lem:stein} without jumps.
We show that $\sbar_{t,\epsilon}\to\sbar_t$ and $\bar b_{t,\epsilon}\to\bar b_t$ as $\epsilon\to0$, by induction from $\sbar_{1,\epsilon}=\sbar_1$. Let $G_\epsilon(\sigma)=\E[(\eta_{t,\epsilon}(\bar x_\sigma)-x)^2]$, with $\bar x_\sigma=x+\sigma\zt$, and define $G(\sigma)$ likewise with $\eta_t$.
Since the two integrands differ only if $\bar x_\sigma\in\cW_{t,\epsilon}$, Lemma~\ref{lem:ramp}, Cauchy--Schwarz and \eqref{eq:tube} yield $|G_\epsilon(\sigma)-G(\sigma)|\le C\sqrt\epsilon$ uniformly over $\sigma\in[\sigma_z,\sigma_{\max}]$, where $\sigma_{\max}<\infty$ bounds $\sbar_t$ and $\sbar_{t,\epsilon}$ for $t\le T+1$ and $\epsilon<g/2$, since $\eta_t$ grows at most linearly and $|\eta_{t,\epsilon}|\le|\eta_t|+\Delta'$.
In turn, $G$ is continuous at $\sbar_t$ by dominated convergence, as $\eta_t$ is continuous off $\cJ_t$, such that $\sbar^2_{t+1,\epsilon}=\sigma_z^2+\alpha G_\epsilon(\sbar_{t,\epsilon})\to\sbar^2_{t+1}$.
The same argument with $\zt\eta_{t,\epsilon}(\bar x_\sigma)$ and $\zt\eta_t(\bar x_\sigma)$ yields $\bar b_{t,\epsilon}\to\bar b_t$.

\emph{Step 2 (Ramped \ac{AMP}):} By Lemmas~\ref{lem:bm} and~\ref{lem:ramp}, a.s., for $t\le T+1$, the empirical distribution of $\{(x_k,\bar x^t_{\epsilon,k})\}_{k=1}^K$ converges in Wasserstein-2 distance to the law of $(x,x+\sbar_{t,\epsilon}\zt)$, $\norm{\br^t_\epsilon}^2/N\to\sbar^2_{t,\epsilon}$, and $b_{t,\epsilon}\to\bar b_{t,\epsilon}$ for $t\le T$.
By the portmanteau theorem \cite[Thm.~2.1]{billingsley1999convergence} for the closed sets $\cJ_{t,s}$ and $\cW_{t,\epsilon}$, and by \eqref{eq:tube}, $\limsup_K\nu_{\bxb^t_\epsilon}(\cJ_{t,s})\le2Jsf_{\max}$ and $\limsup_K\nu_{\bxb^t_\epsilon}(\cW_{t,\epsilon})\le2J\epsilon f_{\max}$.
We work on the event of probability one on which these statements hold for all $\epsilon\in\mathcal E$ and all rational $s>0$, hence for all $s>0$, and on which $\norm{\bH}_{\rm op}\to1+\sqrt\alpha$ \cite{geman1980limit}.

\emph{Step 3 (One-step comparison):} Fix $\epsilon\in\mathcal E$, let $e^t=\norm{\bxb^t-\bxb^t_\epsilon}/\sqrt K$, $e_r^t=\norm{\br^t-\br^t_\epsilon}/\sqrt N$ and $e_x^t=\norm{\bx^t-\bx^t_\epsilon}/\sqrt K$, and let $\bar e^t,\bar e_r^t,\bar e_x^t$ be their limits superior as $K\to\infty$.
Writing $\bx^t-\bx^t_\epsilon=[\eta_t(\bxb^t)-\eta_t(\bxb^t_\epsilon)]+[\eta_t(\bxb^t_\epsilon)-\eta_{t,\epsilon}(\bxb^t_\epsilon)]$, and bounding the first term by \eqref{eq:flip} with $\bu=\bxb^t_\epsilon$ and the second by Lemma~\ref{lem:ramp} and Step~2, we obtain, for every $s>0$,
\begin{equation}   
(\bar e_x^t)^2\le4L'^2(\bar e^t)^2+4\Delta^2\big[2Jsf_{\max}+(\bar e^t)^2/s^2\big]+4\Delta'^2J\epsilon f_{\max},
\end{equation}
such that $s=(\bar e^t)^{2/3}$, or $s\to0$ if $\bar e^t=0$, yields
\begin{equation}
\label{eq:holder}
(\bar e_x^t)^2\le C_1\big[(\bar e^t)^2+(\bar e^t)^{2/3}+\epsilon\big].
\end{equation}

By \eqref{eq:amp}, $\br^{t+1}-\br^{t+1}_\epsilon=-\bH(\bx^t-\bx^t_\epsilon)+\alpha b_t(\br^t-\br^t_\epsilon)+\alpha(b_t-b_{t,\epsilon})\br^t_\epsilon$ and $\bxb^{t+1}-\bxb^{t+1}_\epsilon=\bx^t-\bx^t_\epsilon+\bH^{\mathsf T}(\br^{t+1}-\br^{t+1}_\epsilon)$, so that, with $\norm{\bH}_{\rm op}\to1+\sqrt\alpha$ and $\norm{\br^t_\epsilon}/\sqrt N\to\sbar_{t,\epsilon}$,
\begin{eqnarray}
\bar e_r^{t+1}\le(1+\sqrt\alpha)\sqrt\alpha\,\bar e_x^t\hspace{-4ex}&&+\alpha\,\limsup_K|b_t|\,\bar e_r^t\\
&&+\alpha\sbar_{t,\epsilon}\limsup_K|b_t-b_{t,\epsilon}|,\nonumber
\end{eqnarray}
\begin{equation}
\bar e^{t+1}\le\bar e_x^t+(1+\sqrt\alpha)\alpha^{-1/2}\,\bar e_r^{t+1}.
\end{equation}

\emph{Step 4 (\ac{SE} from closeness):} Suppose that $\bar e^t,\bar e_r^t\le\phi_t(\epsilon)$ for all $\epsilon\in\mathcal E$, with $\phi_t(\epsilon)\to0$ as $\epsilon\to0$. For PL2 $\psi$ with constant $C$, since $\av{(\bxb^t_\epsilon)^2}\to E_d+\sbar_{t,\epsilon}^2$, Cauchy--Schwarz yields $\limsup_K|\av{\psi(\bx,\bxb^t)}-\av{\psi(\bx,\bxb^t_\epsilon)}|\le C_2(\epsilon)\bar e^t$, with $C_2(\epsilon)=C(1+2\sqrt{E_d}+2(E_d+\sigma_{\max}^2)^{1/2}+\bar e^t)$ bounded over $\mathcal E$. Hence, $\limsup_K|\av{\psi(\bx,\bxb^t)}-\E[\psi(x,x+\sbar_t\zt)]|\le C_2(\epsilon)\phi_t(\epsilon)+|\E[\psi(x,x+\sbar_{t,\epsilon}\zt)]-\E[\psi(x,x+\sbar_t\zt)]|$, which vanishes as $\epsilon\to0$ by Step~1 and the continuity of $\sigma\mapsto\E[\psi(x,x+\sigma\zt)]$, proving \eqref{eq:se_limit}. Likewise, $|\hat\sigma_t-\sbar_t|\le e_r^t+|\norm{\br^t_\epsilon}/\sqrt N-\sbar_{t,\epsilon}|+|\sbar_{t,\epsilon}-\sbar_t|$ yields $\hat\sigma_t\to\sbar_t$.

As in Lemma~\ref{lem:bm}, \eqref{eq:se_limit} for all PL2 $\psi$ gives Wasserstein-2 convergence to $\mu_t$, hence weak convergence with uniformly integrable $\|\mathbf a\|^2$, which extends \eqref{eq:se_limit} to every $\psi$ of quadratic growth with $\mu_t(D_\psi)=0$ \cite[Thms.~2.7, 3.5, 3.6]{billingsley1999convergence}. For $t\le T$, this yields $b_t\to\bar b_t$: trivially in case (a), and in case (b) because $\av{\eta'_{t,{\rm ae}}(\bxb^t)}\to\E[\eta'_{t,{\rm ae}}(x+\sbar_t\zt)]$, the boundary term of \eqref{eq:bcorr} is continuous in $\hat\sigma_t>0$, and \eqref{eq:stein} identifies the limit.

\emph{Step 5 (Induction):} At $t=1$, $\br^1=\br^1_\epsilon=\by$ and $\bxb^1=\bxb^1_\epsilon$, so that $\bar e^1=\bar e^1_r=0$.
If the hypothesis of Step~4 holds at $t\le T$, it yields $b_t\to\bar b_t$, so that $\limsup_K|b_t|=|\bar b_t|$ and $\limsup_K|b_t-b_{t,\epsilon}|=|\bar b_t-\bar b_{t,\epsilon}|$, which vanishes as $\epsilon\to0$ by Steps~1--2. Inserting \eqref{eq:holder} into Step~3 then yields $\bar e^{t+1},\bar e_r^{t+1}\le\phi_{t+1}(\epsilon)$, with $\phi_{t+1}$ explicit in $\phi_t$, $\epsilon$ and $|\bar b_t-\bar b_{t,\epsilon}|$, and $\phi_{t+1}(\epsilon)\to0$. By induction, Step~4 applies at every $t\le T+1$, which proves the \ac{SE} statements, and $\limsup_K\max_{t\le T+1}(e^t+e_r^t)\le2\sum_{t\le T+1}\phi_t(\epsilon)\to0$ proves \eqref{eq:close}.
\end{proof}

\vspace{-2ex}
\section{Numerical Results}\label{sec:numerics}
\vspace{-0.5ex}
In this section, we evaluate slicer-\ac{AMP} with the three coefficients of Table~\ref{tab:coef}, considering \ac{iid} symbols drawn uniformly from $\mathcal M=\{\pm1,\pm3,\dots,\pm(M-1)\}$, scaled to $E_d=1$, such that the \ac{SNR} is $1/\sigma_z^2$, with independent realizations of $(\bH,\bx,\bz)$.
%
%
First, Fig.~\ref{fig:main}(a) verifies Theorem~\ref{thm:se} for 8-\ac{PAM} slicer-\ac{AMP} with $\alpha=\frac12$ and $T=12$ at 20~dB, where the \ac{SE} value of the \ac{MSE} after $12$ iterations is $2.49\cdot10^{-2}$ and $\bar b_t\in[0.71,0.87]$.
The results show that, with $\bar b_t$, the \ac{MSE} gap falls from $1.0\cdot10^{-2}$ at $K=64$ to $4.3\cdot10^{-5}$ at $K=8192$, with a standard error of $1.9\cdot10^{-5}$, decaying like $K^{-1.14}$ (fit for $K\ge64$), while with $\hat b_t$ it falls from $5.3\cdot10^{-3}$ to $3.4\cdot10^{-5}$, like $K^{-1.03}$.
In turn, the \ac{SER} of the slicer decisions on $\bxb^{13}$ approaches its \ac{SE} value $0.1271$ (Corollary~\ref{cor:consequences}), reaching $0.1273$ ($\bar b_t$) and $0.1272$ ($\hat b_t$) at $K=8192$.
In contrast, with the standard coefficient $b_t=0$, the gap is $1.5$--$1.6$ for $K\ge64$, and Fig.~\ref{fig:main}(b) shows that the \ac{MSE} leaves the \ac{SE} trajectory at the second iteration and exceeds $E_d=1$ from $t=6$ on.

%
Next, we compare in Fig.~\ref{fig:main}(c) the \ac{SER} of slicer-\ac{AMP} for 4- and 8-\ac{PAM} with $K=1024$, $N=2048$ and $T=20$.
The results show that, with the standard coefficient, the \ac{SER} stays near $0.59$ (4-\ac{PAM}) and $0.80$ (8-\ac{PAM}) at every \ac{SNR}, while with $\hat b_t$, slicer-\ac{AMP} follows \ac{SE} and approaches Bayes-optimal \ac{AMP}, whose denoiser is the posterior mean at noise level $\hat\sigma_t$, at a cost of $\mathcal O(M)$ operations per symbol\footnote{With $\sbar_t$ in place of $\hat\sigma_t$, this benchmark is covered by Lemma~\ref{lem:bm}.}.
For 4-\ac{PAM} at 18~dB, in particular, the \ac{SER} of slicer-\ac{AMP} is $3.3\cdot10^{-4}$ (\ac{SE}: $3.0\cdot10^{-4}$; Bayes-optimal \ac{AMP}: $3.1\cdot10^{-4}$), and for 8-\ac{PAM} it is $5.1\cdot10^{-4}$ at 24~dB (\ac{SE}: $5.1\cdot10^{-4}$; Bayes-optimal \ac{AMP}: $4.8\cdot10^{-4}$) and $1.0\cdot10^{-5}$ at 26~dB (\ac{SE}: $1.2\cdot10^{-5}$), with no error in $5.12\cdot10^6$ symbols at 28 and 30~dB.

We emphasize, however, that with $\bar b_t$, slicer-\ac{AMP} exhibits an error floor at high \ac{SNR}, near $10^{-2}$ for 8-\ac{PAM}, caused by rare realizations that leave the \ac{SE} trajectory and do not recover,\footnote{A realization fails if over $1\%$ of its symbols are in error.} as $\bar b_t$ follows the \ac{SE} schedule rather than the realization.
For 8-\ac{PAM} at 24~dB, for instance, 7 of 400 realizations fail at $K=1024$ and none of 200 at $K=4096$, where the \ac{SER} is $5.0\cdot10^{-4}$ (\ac{SE}: $5.1\cdot10^{-4}$), which indicates a finite-size effect consistent with Theorem~\ref{thm:se}.

\vspace{-1ex}
\section{Conclusion}
\vspace{-0.5ex}
We offered a rigorous \ac{SE} analysis of \ac{AMP} with discontinuous denoisers, yielding a direct transfer of the Lipschitz \ac{SE} theorem via a crossing bound that keeps the steps of \ac{AMP} close to those of \ac{AMP} on Lipschitz ramps, and an \acs{SE}-corrected Onsager coefficient, computed online from the residual and the known prior.
We showed that \ac{AMP} follows \ac{SE} if its Onsager coefficient {is the distributional or the \acs{SE}-corrected one, both of which contain the boundary term}, enabling hard decisions, finite-level quantizers and hard thresholding as admissible \ac{AMP} denoisers without smoothing, whereas the standard coefficient is inconsistent whenever the boundary term is non-zero.
Compared to tensor programs \cite{yang2019scaling,yang2020tensor}, the argument requires only a finite second moment of the prior, uses Lipschitz \ac{SE} as a black box, and justifies an online coefficient using neither the signal nor the \ac{SE} trajectory. Rates, numbers of iterations that grow with $K$, data-dependent thresholds, non-separable denoisers, and orthogonal and vector \ac{AMP} are left for future work.
\vspace{-2ex}

\balance
\bibliographystyle{IEEEtran}
\bibliography{references}

\end{document}